\documentclass[twocolumn,conference]{IEEEtran}
\usepackage[T1]{fontenc}
\usepackage[latin9]{inputenc}
\usepackage{amsmath}
\usepackage{amsthm}
\usepackage{amssymb}
\usepackage{graphicx}
\usepackage[numbers,sort&compress]{natbib}
\usepackage[unicode=true,
 bookmarks=true,bookmarksnumbered=true,bookmarksopen=true,bookmarksopenlevel=1,
 breaklinks=false,pdfborder={0 0 1},backref=false,colorlinks=false]
 {hyperref}
\hypersetup{pdftitle={Your Title},
 pdfauthor={Your Name},
 pdfborderstyle=,pdfpagelayout=OneColumn,pdfnewwindow=true,pdfstartview=XYZ,plainpages=false}

\makeatletter
\theoremstyle{plain}
\newtheorem{lem}{\protect\lemmaname}
\theoremstyle{plain}
\newtheorem{prop}{\protect\propositionname}
\theoremstyle{remark}
\newtheorem{rem}{\protect\remarkname}

\usepackage[justification=centering]{caption}
\usepackage{cuted}

\renewcommand{\maketag@@@}[1]{\hbox{\m@th\normalsize\normalfont#1}}

\usepackage{etoolbox}
\usepackage{xstring}\usepackage{mfirstuc}\usepackage{textcase}\usepackage[acronym]{glossaries}
\newcommand{\newac}{\newacronym}
\newcommand{\ac}{\gls}

\newcommand{\acpl}{\glspl}

\makeglossaries
\newac{isac}{ISAC}{integrated sensing and communication}
\newac{an}{AN}{artificial noise}
\newac{csi}{CSI}{channel state information}
\newac{csit}{CSIT}{channel state information at transmitter}
\newac{fp}{FP}{fractional planning}
\newac{bme}{BME}{beampattern matching error}
\newac{ula}{ULA}{uniform linear array}
\newac{cu}{CU}{communication user}
\newac{uavs}{UAVs}{unmanned aerial vehicles}
\newac{sinr}{SINR}{signal-to-interference-plus-noise ratio}
\newac{sdr}{SDR}{semi-definite relaxation}
\newac{1d}{1D}{one-dimensional}
\newac{zf}{ZF}{zero-forcing}
\newac{aod}{AoD}{angle of departure}
\newac{los}{LoS}{line-of-sight}
\newac{nlos}{NLoS}{non-line-of-sight}
\newac{bs}{BS}{base station}
\newac{bss}{BSs}{base stations}
\newac{cus}{CUs}{communication users}
\newac{snr}{SNR}{signal-to-noise ratio}
\newac{evd}{EVD}{eigenvalue decomposition}
\newac{qsdp}{QSDP}{quadratic semidefinite programing}
\newac{dof}{DoF}{degrees of freedom}
\newac{soc}{SOC}{second order cone}
\newac{qos}{QoS}{quality-of-service}
\newac{lmi}{LMI}{linear matrix inequality}
\newac{bti}{BTI}{Bernstein-type inequality}
\newac{doa}{DoA}{direction of arrival}
\newac{crb}{CRB}{Cram{\'e}r-Rao bound}
\newac{sdp}{SDP}{semidefinite programing}
\newac{kkt}{KKT}{Karush-Kuhn-Tucker}
\newac{sca}{SCA}{successive convex approximation}

\makeatother

\providecommand{\lemmaname}{Lemma}
\providecommand{\propositionname}{Proposition}
\providecommand{\remarkname}{Remark}

\begin{document}
\title{Fundamental CRB-Rate Tradeoff in Multi-antenna Multicast Channel with
ISAC}
\author{\IEEEauthorblockN{$\mathrm{\textrm{Zixiang Ren}}^{1,2}$, $\textrm{Xianxin Song}^{2}$,
$\textrm{Yuan Fang}^{2}$, Ling Qiu$^{1}$, and Jie Xu$^{2}$}\IEEEauthorblockA{$^{1}$Key Laboratory of Wireless-Optical Communications, Chinese
Academy of Sciences, \\
 School of Information Science and Technology, University of Science
and Technology of China}\IEEEauthorblockA{$^{2}$School of Science and Engineering and Future Network of Intelligence
Institute (FNii),\\
 The Chinese University of Hong Kong, Shenzhen} \IEEEauthorblockA{E-mail: rzx66@mail.ustc.edu.cn, xianxinsong@link.cuhk.edu.cn, fangyuan@cuhk.edu.cn,
lqiu@ustc.edu.cn, xujie@cuhk.edu.cn}}
\maketitle
\begin{abstract}
With technical advancements, how to support common data broadcasting
and sensing at the same time is a new problem towards the next-generation
multiple access (NGMA). This paper studies the multi-antenna multicast
channel with integrated sensing and communication (ISAC), in which
a multi-antenna base station (BS) sends common messages to a set of
single-antenna communication users (CUs) and simultaneously estimates
the parameters of an extended target via radar sensing. Under this
setup, we investigate the fundamental performance tradeoff between
the achievable rate for communication and the estimation Cramér-Rao
bound (CRB) for sensing. First, we derive the optimal transmit covariance
in semi-closed form to maximize the achievable rate while ensuring
the maximum estimation CRB constraint subject to a maximum transmit
power constraint at the BS, and accordingly characterize the outer
bound of the so-called CRB-rate (C-R) region. It is shown that the
optimal transmit covariance should be of full rank, consisting of
both information-carrying and dedicated sensing signals in general.
Next, we consider a practical joint information and sensing beamforming
design, and propose an efficient approach to optimize the joint beamforming
for balancing the C-R tradeoff. Numerical results are presented to
show the C-R region achieved by the optimal transmit covariance and
the joint beamforming, as compared to the benchmark scheme with isotropic
transmission. 
\end{abstract}


\IEEEpeerreviewmaketitle{}

\section{Introduction}

With advancements in webcast and content broadcasting applications,
how to efficiently support common data broadcasting to multiple users
over the so-called multicast channel is becoming an increasingly important
problem in future beyond fifth-generation (B5G) and sixth-generation
(6G) wireless networks. On the other hand, \ac{isac} has emerged
as a promising technique to enable dual-functional B5G and 6G wireless
networks, which provide both sensing and communication services \citep{liu2021integrated}.
As a result, how to simultaneously support common data broadcasting
and sensing over the multicast channel is one of the new problems
towards the next-generation multiple access (NGMA).

Recently, the multi-antenna or multiple-input multiple-output (MIMO)
techniques have become an important solution to enhance the \ac{isac}
performance. By equipping multiple antennas at \ac{bs}, MIMO can
exploit the spatial multiplexing and diversity gains to increase the
communication rate and reliability \citep{heath2018foundations},
and provide spatial and waveform diversity gains to enhance the sensing
accuracy and resolution \citep{haimovich2007mimo,LiStoJ07}. There
have been various prior works investigating the waveform and beamforming
designs for multi-antenna ISAC. In general, there are several waveform
design approaches, namely sensing-centric, communication-centric,
and unified waveform designs \citep{liu2021integrated}. Among them,
utilizing properly designed unified waveforms is particularly promising
to maximize the ISAC performance with enhanced spectrum utilization
efficiency. For instance, the authors in \citep{LiuZhouJ18,LiuHuangNirJ20,hua2021optimal}
presented the transmit beamforming in downlink multiuser ISAC systems
over a broadcast channel, in order to optimize the transmit beampattern
for sensing and the \ac{sinr} for communication. The authors in
\citep{NOMAISAC} investigated the transmit beamforming design in
a broadcast channel for ISAC with non-orthogonal multiple access (NOMA).
In addition, \citep{yin2022rate} investigated the C-R tradeoff for
ISAC in multi-antenna broadcast channels with the emerging rate-splitting
multiple access (RSMA) technique.

How to characterize the sensing and communication performance limits
from estimation theory and information theory perspectives is a fundamental
question in \ac{isac} systems (see, e.g., \citep{liu2021integrated}).
While the channel capacity serves as the communication rate limits
(upper bound), the \ac{crb} can act as the sensing performance
limits for target parameters estimation, by providing the variance
lower bound of any unbiased estimators \citep{liu2022survey,liu2021cramer,xiong2022flowing,Haocheng2022}.
Therefore, understanding the \ac{crb}-rate (C-R) tradeoff is an
important problem to reveal the fundamental ISAC limits. For instance,
the authors in \citep{liu2021cramer} optimized the CRB in multiuser
broadcasting ISAC systems with transmit beamforming, subject to \ac{sinr}
(or equivalently rate) constraints. \citep{xiong2022flowing} presented
the C-R region for a point-to-point MIMO ISAC system, and \citep{Haocheng2022}
characterized the whole Pareto boundary of the C-R region for a MIMO
ISAC system with an extended target.

Different from prior works studying the multi-antenna ISAC over point-to-point
and broadcast channels, this paper investigates the multi-antenna
ISAC over a multicast channel towards NGMA, in which a multi-antenna
\ac{bs} sends common messages to a set of single-antenna \acpl{cu}
and simultaneously uses the echo messages to estimate an extended
sensing target. To our best knowledge, how to characterize the fundamental
capacity and C-R tradeoff for the multicast channel with ISAC has
not been studied in the literature yet, thus motivating our current
work.

In particular, we characterize the Pareto boundary of the C-R region
for the new multi-antenna multicast ISAC system, and present practical
joint beamforming design. First, we define the C-R region as the set
of the estimation CRB and the multicast rate pairs that can be simultaneously
achieved by the \ac{isac} system, and obtain two boundary points
corresponding to CRB minimization and rate maximization, respectively.
Then, to characterize the complete Pareto boundary, we present a new
CRB-constrained multicast rate maximization problem, and derive the
optimal covariance solution in semi-closed form by applying the Lagrange
duality methods. It is shown that the optimal covariance should be
of full rank, which consists of both information-carrying and dedicated
sensing signals in general. Furthermore, we also consider practical
joint communication and sensing beamforming designs for multicast
ISAC, and develop an efficient algorithm based on \ac{sca} to find
a high-quality joint beamforming solution to balance the C-R tradeoff.
Finally, we provide numerical results to show the achievable C-R regions
by the optimal covariance and transmit beamforming, as compared to
the benchmark scheme with isotropic transmission.

\textit{Notations}: Vectors and matrices are denoted by bold lower-
and upper-case letters, respectively. $\mathbb{C}^{N\times M}$ denotes
the space of $N\times M$ complex matrices. $\boldsymbol{I}$ and
$\boldsymbol{0}$ represent an identity matrix and an all-zero matrix
with appropriate dimensions, respectively. For a square matrix $\boldsymbol{A}$,
$\textrm{tr}(\boldsymbol{A})$ denotes its trace, and $\boldsymbol{A}\succeq\boldsymbol{0}$
means that $\boldsymbol{A}$ is positive semi-definite. For a complex
arbitrary-size matrix $\boldsymbol{B}$, $\textrm{rank}(\boldsymbol{B})$,
$\boldsymbol{B}^{T}$, $\boldsymbol{B}^{H}$, and $\boldsymbol{B}^{c}$
denote its rank, transpose, conjugate transpose, and complex conjugate,
respectively. $\mathbb{E}(\cdot)$ denotes the stochastic expectation,
$\|\cdot\|$ denotes the Euclidean norm of a vector, and $|\cdot|$
and $\mathrm{Re}(\cdot)$ denote the absolute value and the real component
of a complex entry. $\mathcal{CN}(\boldsymbol{x},\boldsymbol{Y})$
denotes a circularly symmetric complex Gaussian (CSCG) random vector
with mean vector $\boldsymbol{x}$ and covariance matrix $\boldsymbol{Y}$.
$\boldsymbol{A}\otimes\boldsymbol{B}$ represents the Kronecker product
of two matrices $\boldsymbol{A}$ and $\boldsymbol{B}$.

\section{System Model}

We consider an \ac{isac} system over a multicast channel, in which
a BS sends common messages to $K>1$ CUs indexed by $\mathcal{K}\overset{\triangle}{=}\{1,\dots,K\}$
and uses the echo signals to estimate an extended sensing target.
Suppose that the BS is equipped with $N_{t}>1$ transmit antennas
and $N_{r}\ge N_{t}$ receive antennas, and each CU is equipped with
a single antenna.

First, we consider the ISAC signal transmission at the BS. Let $\boldsymbol{x}(n)\in\mathbb{C}^{N_{t}\times1}$
denote the transmitted unified signal for ISAC at symbol $n$, which
is assumed to be an independent CSCG random vector with mean $\boldsymbol{0}$
and covariance matrix $\boldsymbol{S}_{x}=\mathbb{E}(\boldsymbol{x}(n)\boldsymbol{x}^{H}(n))$,
i.e., $\boldsymbol{x}(n)\sim\mathcal{CN}\big(\boldsymbol{0},\boldsymbol{S}_{x}\big),\forall n$.
Suppose that the maximum transmit power budget is $P$. We have the
transmit power constraint as
\begin{equation}
\mathrm{tr}(\boldsymbol{S}_{x})\leq P.\label{eq:Average power}
\end{equation}

Next, we consider the multicast channel for communication. Let $\boldsymbol{h}_{k}\in\mathbb{C}^{N_{t}\times1}$
denote the channel vector from the BS to CU $k\in\mathcal{K}$. The
received signal at the receiver of \ac{cu} $k\in\mathcal{K}$ is
given by
\begin{equation}
y_{k}(n)=\boldsymbol{h}_{k}^{H}\boldsymbol{x}(n)+z_{k}(n),\label{eq:Received com signal}
\end{equation}
where $z_{k}(n)$ denotes the noise at the receiver of CU $k$ that
is a CSCG random variable with zero mean and variance $\sigma^{2}$,
i.e., $z_{k}(n)\sim\mathcal{CN}(0,\sigma^{2}),\forall k\in\mathcal{K}$.
We assume quasi-static channel models, in which the channel vectors
$\{\boldsymbol{h}_{k}\}$ remain unchanged over the transmission blocks
of interest. In order to characterize the fundamental performance
limits, we assume that the BS perfectly knows the global channel state
information of $\{\boldsymbol{h}_{k}\}_{k=1}^{K}$, and each CU $k$
perfectly knows the local CSI of $\boldsymbol{h}_{k}$. Based on the
received signal in (\ref{eq:Received com signal}), the received \ac{snr}
at \ac{cu} $k\in\mathcal{K}$ is
\begin{equation}
\gamma_{k}=\mathbb{E}\bigg(\frac{\big|\boldsymbol{h}_{k}^{H}\boldsymbol{x}(n)\big|^{2}}{\big|z_{k}(n)\big|^{2}}\bigg)=\frac{\boldsymbol{h}_{k}^{H}\boldsymbol{S}_{x}\boldsymbol{h}_{k}}{\sigma^{2}}.\label{eq:SNR for com}
\end{equation}
Accordingly, the achievable rate of the multicast channel \citep{jindal2006capacity}
with given transmit covariance $\boldsymbol{S}_{x}$ is given by
\begin{equation}
R(\boldsymbol{S}_{x})=\underset{k\in\mathcal{K}}{\min}\log_{2}\Big(1+\frac{\boldsymbol{h}_{k}^{H}\boldsymbol{S}_{x}\boldsymbol{h}_{k}}{\sigma^{2}}\Big).\label{eq:rate}
\end{equation}
Then, we consider the radar sensing for estimating an extended target.
We focus on a particular radar processing interval with a total of
$L$ symbols. The extended target is modeled as a surface with $M$
distributed point-like scatters \citep{liu2021cramer}. The angle
of arrival or departure (AoA/AoD) of the $m$-th scatter is denoted
by $\theta_{m},m\in\{1,\dots,M\}$. The target response matrix $\boldsymbol{G}\in\mathbb{C}^{N_{r}\times N_{t}}$
is
\begin{equation}
\boldsymbol{G}=\sum_{m=1}^{M}\beta_{m}\boldsymbol{a}_{r}^{c}(\theta_{m})\boldsymbol{a}_{t}^{H}(\theta_{m}),
\end{equation}
where $\{\beta_{m}\}$ denote complex amplitudes proportional to the
radar cross sections (RCSs) of scatterers, and $\boldsymbol{a}_{r}(\theta)$
and $\boldsymbol{a}_{t}(\theta)$ denote the receive and transmit
steering vectors with angle $\theta$, respectively. Let $\boldsymbol{X}=[\boldsymbol{x}(1),\dots,\boldsymbol{x}(L)]$
denote the transmitted signal over the $L$ symbols. By assuming that
$L$ is fixed and sufficiently large, the sample coherence matrix
of $\boldsymbol{X}$ can be approximated as the covariance matrix
$\boldsymbol{S}_{x}$\footnote{The approximation of the sample coherence matrix $\frac{1}{L}\boldsymbol{X}\boldsymbol{X}^{H}$
as the covariance matrix $\boldsymbol{S}_{x}$ has been widely adopted
in the ISAC literature (see, e.g., \citep{liu2021cramer,LiuHuangNirJ20}).
Such approximation has been shown to be sufficiently accurate in \citep{StoPETLiJ07}
when the sample length is $L=256$, and it is expected to be more
accurate when $L$ becomes larger. }, i.e.,
\begin{equation}
\frac{1}{L}\boldsymbol{X}\boldsymbol{X}^{H}\approx\boldsymbol{S}_{x}.\label{eqn:app}
\end{equation}
In this case, the received signal $\boldsymbol{Y}\in\mathbb{C}^{N_{r}\times L}$
at the BS over the $L$ symbols is\citep{StoPETLiJ07}
\begin{equation}
\boldsymbol{Y}=\boldsymbol{G}\boldsymbol{X}+\boldsymbol{Z},\label{eq:received signal}
\end{equation}
where $\boldsymbol{Z}\in\mathbb{C}^{N_{r}\times L}$ denotes the noise
term, each element of which is an independent CSCG random variable
with zero mean and variance $\sigma_{r}^{2}$. For the general extended
target, we choose the target response matrix $\boldsymbol{G}\in\mathbb{C}^{N_{r}\times N_{t}}$
as the parameter to be estimated. To obtain the CRB for estimating
$\boldsymbol{G}$, we define $\tilde{\boldsymbol{g}}=\mathrm{vec}(\boldsymbol{G})\in\mathbb{C}^{N_{r}N_{t}\times1}$,
and accordingly express the signal model in \eqref{eq:received signal}
as the following complex classical linear model \citep{kay1993fundamentals}:
\begin{equation}
\tilde{\boldsymbol{y}}=(\boldsymbol{X}^{T}\otimes\boldsymbol{I}_{N_{r}})\tilde{\boldsymbol{g}}+\tilde{\boldsymbol{z}},
\end{equation}
where $\tilde{\boldsymbol{y}}=\mathrm{vec}(\boldsymbol{Y})\in\mathbb{C}^{N_{r}L\times1}$,
and $\tilde{\boldsymbol{z}}=\mathrm{vec}(\boldsymbol{Z})\in\mathbb{C}^{N_{r}L\times1}$.
Hence, the received signal vector is a complex Gaussian random vector,
i.e., $\tilde{\boldsymbol{y}}\sim\mathcal{CN}\big((\boldsymbol{X}^{T}\otimes\boldsymbol{I})\tilde{\boldsymbol{g}},\sigma_{r}^{2}\boldsymbol{I}\big)$.
It has been established in \citep{kay1993fundamentals} that the \ac{crb}
matrix for estimating $\tilde{\boldsymbol{g}}$ is
\begin{align}
\boldsymbol{C}=~ & \big((\boldsymbol{X}^{T}\otimes\boldsymbol{I}_{N_{r}})^{H}(\sigma_{r}^{2}\boldsymbol{I})^{-1}(\boldsymbol{X}^{T}\otimes\boldsymbol{I}_{N_{r}})\big)^{-1}\nonumber \\
=~ & \sigma_{r}^{2}\big(\boldsymbol{X}^{c}\boldsymbol{X}^{T}\otimes\boldsymbol{I}_{N_{r}}\big)^{-1}\overset{(\mathrm{a})}{=}\frac{\sigma_{r}^{2}}{L}\big(\boldsymbol{S}_{x}^{T}\otimes\boldsymbol{I}_{N_{r}}\big)^{-1},
\end{align}
where (a) follows from \eqref{eqn:app}. Based on the CRB matrix $\boldsymbol{C}$
, we use trace of the \ac{crb} matrix as the performance metric
for the estimation of $\boldsymbol{G}$ \citep{liu2021cramer}, i.e.,
\begin{equation}
\mathrm{CRB}(\boldsymbol{S}_{x})=\frac{N_{r}\sigma_{r}^{2}}{L}\mathrm{tr}(\boldsymbol{S}_{x}^{-1}).\label{eq:CRB}
\end{equation}
It is assumed that $\{\boldsymbol{h}_{k}\}$ are perfectly known by
the BS. Our objective is to optimize the transmit covariance $\boldsymbol{S}_{x}$
to balance the tradeoff between the achievable rate $R(\boldsymbol{S}_{x})$
in (\ref{eq:rate}) and the CRB $\mathrm{CRB}(\boldsymbol{S}_{x})$
in (\ref{eq:CRB}).

\section{C-R region Characterization}

This section defines the C-R region of the multicast channel with
ISAC, and then characterizes the Pareto boundary of this region that
achieves the optimal C-R tradeoff.

\subsection{C-R Region}

To start with, we define the C-R region, which corresponds to the
set of all rate and \ac{crb} pairs that can be simultaneously achieved
by this system under the maximum power budget $P$, i.e.,
\begin{align}
\mathcal{C}(P)\overset{\triangle}{=}~
\big\{(\hat{\Gamma},\hat{R})| & \hat{\Gamma}\geq\mathrm{CRB}(\boldsymbol{S}_{x}),\hat{R}\leq R(\boldsymbol{S}_{x}),\nonumber \\
 & \mathrm{tr}(\boldsymbol{S}_{x})\leq P,\boldsymbol{S}_{x}\succeq\boldsymbol{0}\big\}.\label{eq:C-R region}
\end{align}
To optimally balance the C-R tradeoff, we characterize the whole Pareto
boundary of region $\mathcal{C}(P)$ in (\ref{eq:C-R region}). Towards
this end, we first derive two boundary points corresponding to the
maximum rate and the minimum CRB, respectively.

First, we consider the rate maximization problem of the multicast
channel:
\begin{eqnarray}
 & \underset{\boldsymbol{S}_{x}\succeq\boldsymbol{0}}{\max} & \underset{k\in\mathcal{K}}{\min}~\log_{2}\Big(1+\frac{\boldsymbol{h}_{k}^{H}\boldsymbol{S}_{x}\boldsymbol{h}_{k}}{\sigma^{2}}\Big)\nonumber \\
 & \mathrm{s.t.} & \mathrm{tr}(\boldsymbol{S}_{x})\leq P.\label{eq:capacity problem}
\end{eqnarray}
It has been shown in \citep{jindal2006capacity} that problem (\ref{eq:capacity problem})
is optimally solvable via the technique of semidefinite programming
(SDP). Let $\boldsymbol{S}_{x}^{\mathrm{com}}$ denote the optimal
solution to problem (\ref{eq:capacity problem}), for which the maximum
achievable rate is $R_{\mathrm{max}}=R(\boldsymbol{S}_{x}^{\mathrm{com}})$.
Accordingly, the achievable CRB is $\mathrm{CRB_{com}}=\frac{N_{r}\sigma_{r}^{2}}{L}\mathrm{tr}({\boldsymbol{S}_{x}^{\mathrm{com}}}^{-1})$.
Notice that if $\boldsymbol{S}_{x}^{\mathrm{com}}$ is rank deficient,
we have $\mathrm{CRB_{com}}\rightarrow\infty$, which means the transmit
degrees of freedom (DoF) are not sufficient to estimate the target
response matrix $\boldsymbol{G}$. By contrast, if $\boldsymbol{S}_{x}^{\mathrm{com}}$
is full rank, then we can obtain a finite $\mathrm{CRB_{com}}$. The
corresponding rate-maximization boundary point is obtained as $(R_{\mathrm{max}},\mathrm{CRB_{com}})$.

Next, we consider the CRB minimization problem:
\begin{eqnarray}
 & \underset{\boldsymbol{S}_{x}\succeq\boldsymbol{0}}{\min} & \frac{N_{r}\sigma_{r}^{2}}{L}\mathrm{tr}(\boldsymbol{S}_{x}^{-1})\nonumber \\
 & \mathrm{s.t.} & \mathrm{tr}(\boldsymbol{S}_{x})\leq P.\label{eq:minimum crb}
\end{eqnarray}
 It has been shown in \citep{liu2021cramer} that the optimal solution
to (\ref{eq:minimum crb}) is $\boldsymbol{S}_{x}^{\mathrm{sen}}=\frac{P}{N_{t}}\boldsymbol{I}$,
i.e., isotropic transmission is optimal. Accordingly, the minimum
\ac{crb} is obtained as $\mathrm{CRB_{min}}=\frac{N_{r}N_{t}^{2}\sigma_{r}^{2}}{LP}$,
and the achievable multicast rate is given as $R_{\mathrm{sen}}=\log_{2}(1+P/N_{t}\underset{k\in\mathcal{K}}{\min}\|\boldsymbol{h}_{k}\|^{2}/\sigma^{2})$.
The corresponding CRB-minimization boundary point is obtained as $(R_{\mathrm{sen}},\mathrm{CRB_{min}})$.

After finding the rate-maximization and the CRB-minimization boundary
points $(R_{\mathrm{max}},\mathrm{CRB_{com}})$ and $(R_{\mathrm{sen}},\mathrm{CRB_{min}})$,
it only remains to obtain the remaining boundary points between them
for characterizing the whole Pareto boundary of the C-R region. Towards
this end, we formulate the following \ac{crb} constrained rate
maximization problem (P1): \vspace{-0.5cm}
\begin{subequations}
\begin{eqnarray}
(\mathrm{P1}): & \underset{\boldsymbol{S}_{x}\succeq\boldsymbol{0}}{\max} & \underset{k\in\mathcal{K}}{\min}~\log_{2}\Big(1+\frac{\boldsymbol{h}_{k}^{H}\boldsymbol{S}_{x}\boldsymbol{h}_{k}}{\sigma^{2}}\Big)\nonumber \\
 & \mathrm{s.t.} & \frac{N_{r}\sigma_{r}^{2}}{L}\mathrm{tr}(\boldsymbol{S}_{x}^{-1})\leq\bar{\Gamma}\label{eq:CRB constraint}\\
 &  & \mathrm{tr}(\boldsymbol{S}_{x})\leq P,\label{eq:Power constraint}
\end{eqnarray}
\end{subequations}
 where $\bar{\Gamma}$ denotes the maximum \ac{crb} threshold that
is set between $\mathrm{CRB_{min}}$ and $\mathrm{CRB_{com}}$. Suppose
that $\bar{R}$ corresponds the optimal objective value of problem
(P1) under a given $\bar{\Gamma}$, and then $(\bar{\Gamma},\bar{R})$
corresponds to one Pareto boundary point. By exhausting $\bar{\Gamma}$
between $\mathrm{CRB_{min}}$ and $\mathrm{CRB_{com}}$, we can obtain
the whole boundary of the C-R region. Notice that problem (P1) is
a convex optimization problem. In the following subsection, we derive
its semi-closed-form solution by using the Lagrange duality method
\citep{boyd2004convex}.

\subsection{Optimal Semi-Closed-Form Solution to Problem (P1)}

To solve problem (P1), we introduce an auxiliary variable $t$ and
define $\Gamma=\frac{\bar{\Gamma}L}{N\sigma_{r}^{2}}$. Accordingly,
problem (P1) is reformulated as 
\begin{subequations}
\begin{eqnarray}
(\mathrm{P1.1}): & \underset{\boldsymbol{S}_{x}\succeq\boldsymbol{0},t}{\min} & -t\nonumber \\
 & \mathrm{s.t.} & \boldsymbol{h}_{k}^{H}\boldsymbol{S}_{x}\boldsymbol{h}_{k}\geq t,\forall k\in\mathcal{K}\label{eq:SNR constraint}\\
 &  & \mathrm{tr}(\boldsymbol{S}_{x}^{-1})\leq\Gamma\label{eq:crb constraint}\\
 &  & \mathrm{tr}(\boldsymbol{S}_{x})\leq P.\label{eq:Power constraint2}
\end{eqnarray}
\end{subequations}
 As the objective function of (P1.1) is convex and all the constraints
are convex, problem (P1.1) is convex. Furthermore, it is easy to show
that (P1.1) satisfies the Slater's conditions \citep{boyd2004convex}.
As a result, the strong duality holds between problem (P1.1) and its
dual problem. Therefore, we use the Lagrange duality method to find
the optimal solution to (P1.1) in well structures. Let $\{\mu_{k}\geq0\}$,
$\lambda_{1}\ge0$, and $\lambda_{2}\ge0$ denote the dual variables
associated with the constraints (\ref{eq:SNR constraint}), (\ref{eq:crb constraint}),
and (\ref{eq:Power constraint2}), respectively. Then, the Lagrangian
of problem (P1.1) is
\begin{align}
 & \mathcal{L}(\boldsymbol{S}_{x},t,\lambda_{1},\lambda_{2},\{\mu_{k}\})\nonumber \\
=~ & t(\sum_{k=1}^{K}\mu_{k}-1)+\mathrm{tr}\big((\lambda_{2}\boldsymbol{I}-\sum_{k=1}^{K}\mu_{k}\boldsymbol{h}_{k}\boldsymbol{h}_{k}^{H})\boldsymbol{S}_{x}\big)\nonumber \\
 & +\lambda_{1}\mathrm{tr}(\boldsymbol{S}_{x}^{-1})-\lambda_{1}\Gamma-\lambda_{2}P.\label{eq:Rewrite dual function}
\end{align}
Accordingly, the dual function is given by
\begin{eqnarray}
g\big(\lambda_{1},\lambda_{2},\{\mu_{k}\}\big)= & \underset{\boldsymbol{S}_{x}\succeq\boldsymbol{0},t}{\min} & \mathcal{L}(\boldsymbol{S}_{x},t,\lambda_{1},\lambda_{2},\{\mu_{k}\}).\vspace{-0.5cm}\label{eq:dual function-1}
\end{eqnarray}
In order for $g\big(\lambda_{1},\lambda_{2},\{\mu_{k}\}\big)$ to
be bounded from below, it must hold that $\sum_{k=1}^{K}\mu_{k}=1$
and $\boldsymbol{A}(\lambda_{2},\{\mu_{k}\})\overset{\triangle}{=}\lambda_{2}\boldsymbol{I}-\sum_{k=1}^{K}\mu_{k}\boldsymbol{h}_{k}\boldsymbol{h}_{k}^{H}\succeq\boldsymbol{0}$.
Therefore, the dual problem of (P1.1) is given as 
\begin{subequations}
\begin{align}
\textrm{(D1.1)}:\underset{\lambda_{1},\lambda_{2},\{\mu_{k}\}}{\max} & g\big(\lambda_{1},\lambda_{2},\{\mu_{k}\}\big)\nonumber \\
\mathrm{s.t.} & \sum_{k=1}^{K}\mu_{k}=1\label{eq:sum mu}\\
 & \boldsymbol{A}(\lambda_{2},\{\mu_{k}\})\succeq\boldsymbol{0}\label{eq:Amatrix}\\
 & \lambda_{1}\geq0,\lambda_{2}\geq0,\mu_{k}\geq0,\forall k\in\mathcal{K}.\label{eq:nonnegative}
\end{align}
\end{subequations}
 For notational convenience, let $\mathcal{D}$ denote the feasible
region of $\lambda_{1}$, $\lambda_{2}$, and $\{\mu_{k}\}$ characterized
by (\ref{eq:sum mu}), (\ref{eq:Amatrix}), and (\ref{eq:nonnegative}).
Let $\boldsymbol{S}_{x}^{*}$, and $t^{*}$ denote the optimal solution
to problem \eqref{eq:dual function-1} with given $(\lambda_{1},\lambda_{2},\{\mu_{k}\})\in\mathcal{D}$.
Furthermore, let $\boldsymbol{S}_{x}^{\star}$ and $t^{\star}$ denote
the optimal primal solution to problem (P1.1), and $\lambda_{1}^{\star}$,
$\lambda_{2}^{\star},$ and $\{\mu_{k}^{\star}\}$ denote the optimal
dual variables to problem (D1.1).

As the strong duality holds between problem (P1.1) and its dual problem
(D1.1), problem (P1.1) can be solved by equivalently solving the dual
problem (D1.1). In the following, we first solve problem \eqref{eq:dual function-1}
with given $(\lambda_{1},\lambda_{2},\{\mu_{k}\})\in\mathcal{D}$,
then find $\lambda_{1}^{\star}$, $\lambda_{2}^{\star},$ and $\{\mu_{k}^{\star}\}$
for problem (D1.1), and finally obtain the optimal primary solution
$\boldsymbol{S}_{x}^{\star}$ and $t^{\star}$ to problem (P1.1).

\subsubsection{Optimal Solution to (\ref{eq:dual function-1}) to Obtain $g\big(\lambda_{1},\lambda_{2},\{\mu_{k}\}\big)$}

First, we evaluate the dual function $g\big(\lambda_{1},\lambda_{2},\{\mu_{k}\}\big)$
with any given $(\lambda_{1},\lambda_{2},\{\mu_{k}\})\in\mathcal{D}$.
To this end, we suppose that $\mathrm{rank}(\boldsymbol{A}(\lambda_{2},\{\mu_{k}\}))=N\le N_{t}$.
Accordingly, we express the \ac{evd} of $\boldsymbol{A}(\lambda_{2},\{\mu_{k}\})$
as $\boldsymbol{A}(\lambda_{2},\{\mu_{k}\})=\boldsymbol{U}\boldsymbol{\Lambda}\boldsymbol{U}^{H}$,
where $\boldsymbol{U}\boldsymbol{U}^{H}=\boldsymbol{U}^{H}\boldsymbol{U}=\boldsymbol{I}$,
and $\boldsymbol{\Lambda}=\mathrm{diag}(\alpha_{1},\dots,\alpha_{N_{t}})$
with $\alpha_{1}\ge\dots\ge\alpha_{N_{t}}$ being the eigenvalues
of $\boldsymbol{A}(\lambda_{2},\{\mu_{k}\})$. We then have the following
Lemma. 
\begin{lem}
\label{lemma1} \textup{\label{lem:1}For any given $(\lambda_{1},\lambda_{2},\{\mu_{k}\})\in\mathcal{D}$,
we obtain the optimal solution $\boldsymbol{S}_{x}^{*}$ to problem
(\ref{eq:dual function-1}) by considering the following two cases:}
\end{lem}
\begin{itemize}
\item When $\lambda_{1}=0$, any $\boldsymbol{S}_{x}^{*}$ satisfying $\boldsymbol{A}(\lambda_{2},\{\mu_{k}\})\boldsymbol{S}_{x}^{*}=\boldsymbol{0}$
(or lying in the null space of $\boldsymbol{A}(\lambda_{2},\{\mu_{k}\})$
is optimal to problem (\ref{eq:dual function-1}).\footnote{In this case, we choose $\boldsymbol{S}_{x}^{*}=\boldsymbol{0}$ to
evaluate dual function $g\big(\lambda_{1},\lambda_{2},\{\mu_{k}\}\big)$
only, though it is generally infeasible for primal problem (P1.1).} 
\item When $\lambda_{1}>0$, we have
\begin{align}
\boldsymbol{S}_{x}^{*}=\lambda_{1}^{1/2}\boldsymbol{U}{\boldsymbol{\Sigma}}\boldsymbol{U}^{H},\label{eqn:opt:Sx}
\end{align}
where ${\boldsymbol{\Sigma}}=\mathrm{diag}(\beta_{1},\dots,\beta_{N_{t}})$
with $\beta_{i}=\alpha_{i}^{-1/2}$, for $i\leq N$ and $\beta_{i}\rightarrow+\infty$,
for $N<i\leq N_{t}$. 
\end{itemize}
\begin{proof} Please refer to Appendix A. \end{proof}

Based on Lemma \ref{lemma1}, dual function $g\big(\lambda_{1},\lambda_{2},\{\mu_{k}\}\big)$
is obtained.

\subsubsection{Optimal Solution to Dual Problem (D1.1)}

Next, we find the optimal dual variables $\lambda_{1}^{\star}$, $\lambda_{2}^{\star},$
and $\{\mu_{k}^{\star}\}$ for solving the dual problem (D1.1). As
(D1.1) is always convex but non-differentiable in general, we use
the subgradient based methods such as the ellipsoid method to find
its optimal solution \citep{boyd2014ellipsoid}. The basic idea of
the ellipsoid method is to first generate a ellipsoid containing $\lambda_{1}^{\star}$,
$\lambda_{2}^{\star},$ and $\{\mu_{k}^{\star}\}$, and then iteratively
construct new ellipsoids containing these variables but with reduced
volumes, until convergence \citep{boyd2014ellipsoid}. To successfully
implement the ellipsoid method, we only need to find the subgradients
of the objective and constraint functions in (D1.1). For notational
convenience, we define $\boldsymbol{H}_{k}=\boldsymbol{h}_{k}\boldsymbol{h}_{k}^{H},\forall k\in\mathcal{K}$. 

First, we remove the equality constraint (\ref{eq:sum mu}) by substituting
$\mu_{K}$ as $\mu_{K}=1-\sum_{k=1}^{K-1}\mu_{k}$ in problem (D1.1).
Next, we consider the objective function, one subgradient of which
at any given $[\mu_{1},\dots,\mu_{K-1},\lambda_{1},\lambda_{2}]^{T}\in\mathbb{C}^{(K+1)\times1}$
is $[\mathrm{tr}\big((\boldsymbol{H}_{1}-\boldsymbol{H}_{K})\boldsymbol{S}_{x}^{*}\big),\dots,\mathrm{tr}\big((\boldsymbol{H}_{K-1}-\boldsymbol{H}_{K})\boldsymbol{S}_{x}^{*}\big),\Gamma-\mathrm{tr}(\boldsymbol{S}_{x}^{*-1}),P-\mathrm{tr}(\boldsymbol{S}_{x}^{*})]^{T}$
\citep{mohseni2006optimized}.

Furthermore, we consider the constraints in \eqref{eq:nonnegative}.
Let $\boldsymbol{e}_{i}\in\mathbb{C}^{K+1}$ denote the vector with
all zero entries except the $i$-th entry being 1. Then, the subgradient
for constraint $\mu_{k}\geq0$ is $-\boldsymbol{e}_{k},k=1,\dots,K-1$.
The subgradients for $\lambda_{1}\geq0$ and $\lambda_{2}\geq0$ are
$-\boldsymbol{e}_{K}$ and $-\boldsymbol{e}_{K+1}$, respectively.
In addition, constraint $\mu_{K}\geq0$ is equivalent to $1-\sum_{k=1}^{K-1}\mu_{k}\geq0$,
whose subgradient is $[\underset{K-1}{\underbrace{1,\dots,1}},0,0]$.

Finally, we consider constraint $\boldsymbol{A}(\lambda_{2},\{\mu_{k}\})\succeq\boldsymbol{0}$
in (\ref{eq:Amatrix}), whose subgradient is given in the following
lemma. 
\begin{lem}
\textup{Let $\boldsymbol{v}\in\mathbb{C}^{N\times1}$ denote the eigenvector
of $\boldsymbol{A}(\lambda_{2},\{\mu_{k}\})$ corresponding to the
smallest eigenvalue. The subgradient of }$\boldsymbol{A}(\lambda_{2},\{\mu_{k}\})\succeq\boldsymbol{0}$\textup{
in (\ref{eq:Amatrix}) at the given }$\big(\lambda_{1},\lambda_{2},\{\mu_{k}\}\big)\in\mathcal{D}$
\textup{is
\[
\begin{array}{c}
[\boldsymbol{v}^{H}(\boldsymbol{H}_{1}-\boldsymbol{H}_{K})\boldsymbol{v},\dots,\boldsymbol{v}^{H}(\boldsymbol{H}_{K-1}-\boldsymbol{H}_{K})\boldsymbol{v},0,-1]^{T}.\end{array}
\]
} 
\end{lem}
\begin{proof} This lemma follows immediately by noting the fact that
$\boldsymbol{A}(\lambda_{2},\{\mu_{k}\})\succeq\boldsymbol{0}\Longleftrightarrow\boldsymbol{v}^{H}\boldsymbol{A}(\lambda_{2},\{\mu_{k}\})\boldsymbol{v}\geq0.$\end{proof} 

So far, the subgradients of the objective function and all constraints
have been obtained. As a result, we can efficiently obtain the optimal
dual solution $\lambda_{1}^{\star}$, $\lambda_{2}^{\star}$, and
$\{\mu_{k}^{\star}\}$ via the ellipsoid method.

\subsubsection{Semi-Closed-Form Solution to Primal Problem (P1.1) or (P1)}

Finally, with $\lambda_{1}^{\star}$, $\lambda_{2}^{\star}$, and
$\{\mu_{k}^{\star}\}$ at hand, we derive the optimal primal solution
$\boldsymbol{S}_{x}^{\star}$ and $t^{\star}$ to problem (P1.1).

Notice that we consider $\bar{\Gamma}$ between $\mathrm{CRB_{min}}\textrm{ and }\mathrm{CRB_{com}}$.
In this case, it can be shown that the optimal dual variables must
satisfy that $\lambda_{1}^{\star}>0$ and $\boldsymbol{A}(\lambda_{2}^{\star},\{\mu_{k}^{\star}\})$
is of full rank (i.e., $\mathrm{rank}(\boldsymbol{A}(\lambda_{2}^{\star},\{\mu_{k}^{\star}\}))=N_{t}$),
since otherwise, the maximum CRB constraint in \eqref{eq:crb constraint}
or the maximum transmit power constraint in (\ref{eq:Power constraint2})
cannot be satisfied. In this case, denote the \ac{evd} of $\boldsymbol{A}(\lambda_{2}^{\star},\{\mu_{k}^{\star}\})$
as $\boldsymbol{A}(\lambda_{2}^{\star},\{\mu_{k}^{\star}\})=\boldsymbol{U}^{\star}\boldsymbol{\Lambda}^{\star}\boldsymbol{U}^{\star H}$,
where $\boldsymbol{\Lambda}^{\star}=\mathrm{diag}(\alpha_{1}^{\star},\dots,\alpha_{N_{t}}^{\star})$
and $\alpha_{1}^{\star}\ge\dots\ge\alpha_{N_{t}}^{\star}>0$ are the
eigenvalues of $\boldsymbol{A}(\lambda_{2}^{\star},\{\mu_{k}^{\star}\})$.
Then the following proposition follows directly from Lemma \ref{lem:1},
for which the proof is skipped for brevity.
\begin{prop}
\textup{\label{thm:1}The optimal} \textup{primal solution $\boldsymbol{S}_{x}^{\star}$
and $t^{\star}$ to problem (P1.1) is given by}
\begin{align}
\boldsymbol{S}_{x}^{\star}= & \lambda_{1}^{\star1/2}\boldsymbol{U}^{\star}{\boldsymbol{\Sigma}^{\star}}\boldsymbol{U}^{\star H},\\
t^{\star}= & \underset{k\in\mathcal{K}}{\min}\ \boldsymbol{h}_{k}^{H}\boldsymbol{S}_{x}^{\star}\boldsymbol{h}_{k},\vspace{-0.6cm}
\end{align}
 \textup{where ${\boldsymbol{\Sigma}^{\star}}=\mathrm{diag}(\beta_{1}^{\star},\dots,\beta_{N_{t}}^{\star}$)
and $\beta_{i}^{\star}=\alpha_{i}^{\star-1/2}$, $\forall i\in\{1,\dots,N_{t}\}$.} 
\end{prop}
As a result, problem (P1.1) or (P1) is finally solved. To gain more
insights, we have the following remark. 
\begin{rem}
With the optimal dual solution $\lambda_{2}^{\star}$ and $\{\mu_{k}^{\star}\}$,
we define the (negative) weighted communication channel of the $K$
CUs as $\boldsymbol{B}=-\sum_{k=1}^{K}\mu_{k}^{\star}\boldsymbol{h}_{k}\boldsymbol{h}_{k}^{H}\succeq\boldsymbol{0}$,
with rank $N_{\text{com}}=\text{rank}(\boldsymbol{B})\le N_{t}$.
The EVD of $\boldsymbol{B}$ is then expressed as $\boldsymbol{B}=[\overline{\boldsymbol{U}}_{\text{com}}~{\boldsymbol{U}}_{\text{com}}]{\boldsymbol{\Delta}}[\overline{\boldsymbol{U}}_{\text{com}}~{\boldsymbol{U}}_{\text{com}}]^{H}$,
where ${\boldsymbol{\Delta}}=\mathrm{diag}(0,\ldots,0,\delta_{1},\ldots,\delta_{N_{\text{com}}})$
with $0>\delta_{1}\ge\ldots\ge\delta_{N_{\text{com}}}$ denoting the
$N_{\text{com}}$ negative eigenvalues, and $\boldsymbol{U}_{\text{com}}\in\mathbb{C}^{N_{t}\times N_{\text{com}}}$
and $\overline{\boldsymbol{U}}_{\text{com}}\in\mathbb{C}^{N_{t}\times(N_{t}-N_{\text{com}})}$
consist the eigenvectors corresponding the non-zero and zero eigenvalues,
thus spanning the communication subspaces and non-communication subspaces,
respectively. In this case, recall that $\boldsymbol{A}(\lambda_{2}^{\star},\{\mu_{k}^{\star}\})=\lambda_{2}^{\star}\boldsymbol{I}-\sum_{k=1}^{K}\mu_{k}^{\star}\boldsymbol{h}_{k}\boldsymbol{h}_{k}^{H}=\lambda_{2}^{\star}\boldsymbol{I}+\boldsymbol{B}\succ\boldsymbol{0}$.
It is evident that the EVD of $\boldsymbol{A}(\lambda_{2}^{\star},\{\mu_{k}^{\star}\})$
can be expressed as
\begin{align*}
\boldsymbol{A}(\lambda_{2}^{\star},\{\mu_{k}^{\star}\})= & \boldsymbol{U}^{\star}\boldsymbol{\Lambda}^{\star}\boldsymbol{U}^{\star H}\\
= & [\overline{\boldsymbol{U}}_{\text{com}}~{\boldsymbol{U}}_{\text{com}}]\big(\lambda_{2}^{\star}\boldsymbol{I}+{\boldsymbol{\Delta}}\big)[\overline{\boldsymbol{U}}_{\text{com}}~{\boldsymbol{U}}_{\text{com}}]^{H},
\end{align*}
where $\boldsymbol{U}^{\star}=[\overline{\boldsymbol{U}}_{\text{com}}~\boldsymbol{U}_{\text{com}}]$
and $\boldsymbol{\Lambda}^{\star}=\mathrm{diag}(\alpha_{1}^{\star},\dots,\alpha_{N_{t}}^{\star})$
with
\[
\alpha_{i}^{\star}=\left\{ \begin{array}{cc}
\lambda_{2}^{\star}, & {\text{if~}}i=1,\ldots,N_{t}-N_{\text{com}}\\
\lambda_{2}^{\star}+\delta_{i-(N_{t}-N_{\text{com}})}, & {\text{if~}}i=N_{t}-N_{\text{com}}+1,\ldots,N_{t}
\end{array}\right..
\]
In this case, it follows from Proposition \ref{thm:1} that the optimal
transmit covariance is given by 
\begin{align*}
\boldsymbol{S}_{x}^{\star} & =\lambda_{1}^{\star1/2}\boldsymbol{U}{\boldsymbol{\Sigma}^{\star}}\boldsymbol{U}^{\star H}\\
 & =\lambda_{1}^{\star1/2}\overline{\boldsymbol{U}}_{\text{com}}\overline{\boldsymbol{\Sigma}}_{\text{com}}^{\star}\overline{\boldsymbol{U}}_{\text{com}}^{H}+\lambda_{1}^{\star1/2}\boldsymbol{U}_{\text{com}}{\boldsymbol{\Sigma}}_{\text{com}}^{\star}\boldsymbol{U}_{\text{com}}^{H},
\end{align*}
where $\overline{\boldsymbol{\Sigma}}_{\text{com}}^{\star}=\mathrm{diag}(\sqrt{1/\lambda_{2}},\ldots,\sqrt{1/\lambda_{2}})$
and ${\boldsymbol{\Sigma}}_{\text{com}}^{\star}=\mathrm{diag}(\sqrt{1/(\lambda_{2}+\delta_{1})},\ldots,\sqrt{1/(\lambda_{2}+\delta_{N_{\text{com}}})})$.
It is clear that the transmit covariance is decomposed into two parts,
including $\lambda_{1}^{\star1/2}\boldsymbol{U}_{\text{com}}{\boldsymbol{\Sigma}}_{\text{com}}^{\star}\boldsymbol{U}_{\text{com}}^{H}$
towards CUs for ISAC, and $\lambda_{1}^{\star1/2}\overline{\boldsymbol{U}}_{\text{com}}\overline{\boldsymbol{\Sigma}}_{\text{com}}^{\star}\overline{\boldsymbol{U}}_{\text{com}}^{H}$
for dedicated for sensing. In the first part for ISAC, more transmit
power is allocated over the subchannels with stronger combined channel
gains (or when the absolute value of negative $\delta_{i}$ becomes
large). In the second part for sensing, equal power allocation is
adopted. 
\end{rem}

\section{Joint Communication and Sensing Beamforming}

The previous section presented the optimal full-rank transmit covariance
for achieving the Perato boundary of the C-R region, in which, however,
the CU receivers may need to implement joint decoding or successive
interference cancellation (SIC) for decoding the information signals.
Alternatively, this section presents a practical joint communication
and sensing beamforming design, in which a single transmit beam is
used for delivering the common message. Let $\boldsymbol{w}\in\mathbb{C}^{N_{t}\times1}$
denote the communication vector, and $x_{\text{com}}(n)$ denote the
common message at symbol $n$ that is a CSCG random variable with
zero mean and unit variance. Besides, let $\boldsymbol{s}(n)$ denote
dedicated sensing signals to provide additional DoF for estimating
the extended target, which is a random vector with mean zero and covariance
matrix $\boldsymbol{S}_{x}$. Then the transmit signal is given by
$\boldsymbol{x}(n)=\boldsymbol{w}x_{\text{com}}(n)+\boldsymbol{s}(n)$,
for which the transmit covariance matrix is $\boldsymbol{S}_{x}=\boldsymbol{S}_{s}+\boldsymbol{w}\boldsymbol{w}^{H}$.
The transmit power constraint becomes $\mathrm{tr}(\boldsymbol{S}_{x})=\mathrm{tr}(\boldsymbol{S}_{s})+\|\boldsymbol{w}\|^{2}\le P$.

First, consider the communication. With joint transmit beamforming,
the dedicated sensing signals may introduce harmful interference at
the receiver of CUs. In this case, the \ac{sinr} at the receiver
of each CU $k\in\mathcal{K}$ is denoted as
\begin{equation}
\tilde{\gamma_{k}}=\frac{|\boldsymbol{h}_{k}^{H}\boldsymbol{w}|^{2}}{\boldsymbol{h}_{k}^{H}\boldsymbol{S}_{s}\boldsymbol{h}_{k}+\sigma^{2}},
\end{equation}
and the corresponding achievable multicast rate is
\begin{equation}
\hat{R}(\boldsymbol{w},\boldsymbol{S}_{s})\overset{\triangle}{=}\underset{k\in\mathcal{K}}{\min}~\log_{2}\Big(1+\frac{|\boldsymbol{h}_{k}^{H}\boldsymbol{w}|^{2}}{\boldsymbol{h}_{k}^{H}\boldsymbol{S}_{s}\boldsymbol{h}_{k}+\sigma^{2}}\Big).
\end{equation}

Next, consider the sensing. As both communication and sensing beams
can be employed for target estimation \citep{liu2021cramer}, the
corresponding estimation CRB is same as \eqref{eq:CRB}.

In this case, the \ac{crb}-constrained rate maximization problem
via joint communication and sensing beamforming is formulated as 
\begin{subequations}
\begin{eqnarray*}
(\mathrm{P2}): & \underset{\boldsymbol{w},\boldsymbol{S}_{s}\succeq\boldsymbol{0}}{\max} & \underset{k\in\mathcal{K}}{\min}\log_{2}\Big(1+\frac{|\boldsymbol{h}_{k}^{H}\boldsymbol{w}|^{2}}{\boldsymbol{h}_{k}^{H}\boldsymbol{S}_{s}\boldsymbol{h}_{k}+\sigma^{2}}\Big)\\
 & \mathrm{s.t.} & \mathrm{tr}(\boldsymbol{S}_{s}+\boldsymbol{w}\boldsymbol{w}^{H})\leq P\\
 &  & \frac{N_{r}\sigma_{r}^{2}}{L}\mathrm{tr}\big((\boldsymbol{S}_{s}+\boldsymbol{w}\boldsymbol{w}^{H})^{-1}\big)\leq\bar{\Gamma}.
\end{eqnarray*}
\end{subequations}
 By adopting an auxiliary variable $t$ and substituting $\boldsymbol{S}_{s}=\boldsymbol{S}_{x}-\boldsymbol{w}\boldsymbol{w}^{H}$,
problem (P2) is equivalently reformulated as\vspace{-0.2cm} 
\begin{subequations}
\begin{align}
(\textrm{P2.1}):\underset{\boldsymbol{w},\boldsymbol{S}_{x},t}{\max} & t\nonumber \\
\mathrm{s.t.} & \frac{|\boldsymbol{h}_{k}^{H}\boldsymbol{w}|^{2}}{\boldsymbol{h}_{k}^{H}\boldsymbol{S}_{x}\boldsymbol{h}_{k}-|\boldsymbol{h}_{k}^{H}\boldsymbol{w}|^{2}+\sigma^{2}}\geq t,\forall k\in\mathcal{K}\label{eq:SINR constraint}\\
 & \textrm{\ensuremath{\frac{N_{r}\sigma_{r}^{2}}{L}}\ensuremath{\mathrm{tr}}(\ensuremath{\boldsymbol{S}_{x}^{-1}})\ensuremath{\leq\bar{\Gamma}}}\label{eq:21}\\
 & \mathrm{tr}(\boldsymbol{S}_{x})\leq P\label{eq:22}\\
 & \boldsymbol{S}_{x}-\boldsymbol{w}\boldsymbol{w}^{H}\succeq\boldsymbol{0}.\label{eq:23}
\end{align}
\end{subequations}

Note that problem (P2.1) is a non-convex optimization problem due
to the non-convex constraints in (\ref{eq:SINR constraint}). To deal
with this issue, we propose an efficient algorithm based on the technique
of \ac{sca} to find a high-quality suboptimal solution. The basic
idea of \ac{sca} is to iteratively approximate the non-convex optimization
problem (P2) into a series of convex problems by linearizing the non-convex
constraint functions in (\ref{eq:SINR constraint}) via the first-order
Taylor approximation, such that each approximate problem can be optimally
solved via standard convex optimization techniques.

The SCA-based algorithm is implemented in an iterative manner. Consider
one particular iteration $i\ge1$, in which the local point of $\boldsymbol{w}$
is denoted by $\boldsymbol{w}^{(i)}$. Note that in \eqref{eq:SINR constraint},
the convex quadratic function $|\boldsymbol{h}_{k}^{H}\boldsymbol{w}|^{2}$
is lower bounded by its first-order Taylor expansion at local point
$\boldsymbol{w}^{(i)}$, i.e., 
\begin{align}
|\boldsymbol{h}_{k}^{H}\boldsymbol{w}|^{2}\le2\mathrm{Re}(\boldsymbol{w}^{H}\boldsymbol{h}_{k}\boldsymbol{h}_{k}^{H}\boldsymbol{w}^{(i)})-\boldsymbol{w}^{(i)H}\boldsymbol{h}_{k}\boldsymbol{h}_{k}^{H}\boldsymbol{w}^{(i)}\triangleq\psi_{k}^{(i)}(\boldsymbol{w}),\label{eqn:lower:bound}
\end{align}
where $\psi_{k}^{(i)}(\boldsymbol{w})$ is an affine function with
respect to $\boldsymbol{w}$. By replacing $|\boldsymbol{h}_{k}^{H}\boldsymbol{w}|^{2}$
as $\psi_{k}^{(i)}(\boldsymbol{w})$, $\forall k\in\mathcal{K}$,
the constraints in (\ref{eq:SINR constraint}) is approximated as
\begin{align}
\frac{\psi_{k}^{(i)}(\boldsymbol{w})}{\boldsymbol{h}_{k}^{H}\boldsymbol{S}_{x}\boldsymbol{h}_{k}-\psi_{k}^{(i)}(\boldsymbol{w})+\sigma^{2}}\geq t,\forall k\in\mathcal{K}\label{eq:first order expansion}
\end{align}
Notice that the LHS of \eqref{eq:first order expansion} serves as
a lower bound of that of \eqref{eq:SINR constraint}, and as a result,
the feasible region of $\boldsymbol{w},\boldsymbol{S}_{x},$ and $t$
characterized by (\ref{eq:first order expansion}) is always a subset
of those by (\ref{eq:SINR constraint}). By replacing (\ref{eq:SINR constraint})
as (\ref{eq:first order expansion}), we obtain the approximate problem
at the $i$-th iteration as 
\begin{subequations}
\begin{eqnarray*}
(\textrm{P2.2.}i): & \underset{\boldsymbol{w},\boldsymbol{S}_{x},t}{\max} & t\\
 & \mathrm{s.t.} & \textrm{(\ref{eq:first order expansion}), (\ref{eq:21}), (\ref{eq:22}), and (\ref{eq:23}).}
\end{eqnarray*}
\end{subequations}
Notice that problem (P2.2.$i$) is a quasi-convex optimization problem,
which can be optimally solved by equivalently solving a series of
convex feasibility problems together with bisection search \citep{boyd2004convex}.
To implement this, we define the feasibility problem associated with
(P2.2.$i$) under given $t$ as 
\begin{subequations}
\begin{eqnarray*}
(\textrm{P2.3.\ensuremath{i}}): & \underset{\boldsymbol{w},\boldsymbol{S}_{x}}{\max} & 0\\
 & \mathrm{s.t.} & \textrm{\textrm{(\ref{eq:first order expansion}), (\ref{eq:21}), (\ref{eq:22}), and (\ref{eq:23})}},
\end{eqnarray*}
\end{subequations}
Notice that the feasibility problem (P2.3.$i$) is a convex optimization
problem, which can thus be solved optimally via standard convex solvers
such as CVX \citep{cvx}. Let the optimal solution of $t$ to problem
(P2.2.$i$) be denoted by $t^{(i)*}$. It is thus clear that if problem
(P2.3.$i$) is feasible under a given $t$, then we have $t^{(i)*}\ge t$;
otherwise, $t^{(i)*}<t$ follows. Therefore, by using bisection search
over $t$ and solving problem (P2.3.$i$) under given $t$'s. Therefore,
problem (P2.2.$i$) has been optimally solved.

Let ${\boldsymbol{w}^{(i)*}}$ denote the optimal solution of ${\boldsymbol{w}}$
to problem (P2.2.$i$), which is then updated as the local point of
${\boldsymbol{w}}$ in the next iteration $i+1$, i.e., ${\boldsymbol{w}^{(i+1)}}\gets{\boldsymbol{w}^{(i)*}}$.
It can be verified that the SCA-based iterations can lead to a monotonically
non-decreasing objective value for problem (P2.1). As a result, the
convergence of the SCA-based algorithm is ensured. After convergence,
the obtained solution to problem (P2.2.$i$) is chosen as the desired
solution, denoted by ${\boldsymbol{w}^{*}}$, ${\boldsymbol{S}^{*}}$,
and $t^{*}$. This thus complete the solution to problem (P2.1) and
thus (P2).

\section{Numerical Results}

This section provides numerical results to show the C-R regions achieved
by the optimal transmit covariance and the suboptimal joint beamforming
design. We consider the following scheme for comparison.
\begin{itemize}
\item \textbf{Isotropic transmission}: The BS sets $\boldsymbol{S}_{x}=\frac{P}{N}\boldsymbol{I}$,
which is same as $\boldsymbol{S}_{x}^{\text{sen}}$ for CRB minimization. 
\end{itemize}
In the simulation, we set the numbers of transmit and receive antennas
at the BS as $N_{t}=N_{r}=4$, the length of radar processing interval
as $L=256$, and the noise power as $\sigma_{r}^{2}=\sigma^{2}=1$.
We set the transmit power $P=0$ dB. The channel vectors of \acpl{cu}
are generated based on normalized Rayleigh fading.

Fig. \ref{fig:2} shows the C-R region in a scenario with $K=35$
CUs. In this case, $\boldsymbol{S}_{x}^{\text{com}}$ is obtained
to be full rank, and as a result, the boundary point $(\text{CRB}_{\text{com}},R_{\text{max}})$
is finite and the time sharing design is applicable. It is observed
that the C-R-region boundary by the optimal transmit covariance dominates
those by other schemes, thus validating the effectiveness of the optimization.

Fig. \ref{fig:3} shows the C-R region in a scenario with $K=3$ CUs.
In this case, $\boldsymbol{S}_{x}^{\text{com}}$ is obtained to be
rank-deficient, such that the boundary point $(\text{CRB}_{\text{com}},R_{\text{max}})$
becomes infinite. It is observed that as the CRB becomes large, the
joint beamforming achieves a C-R-region boundary close to the optimal
transmit covariance and outperforms the isotropic transmission, thus
showing the benefit of beamforming in this case.

Finally, Fig. \ref{fig:4} shows the achievable rate versus the number
$K$ of CUs, where the CRB threshold is set to be $\bar{\Gamma}=0.5$.
When $K$ is small, it is observed that the joint beamforming design
performs close to the optimal transmit covariance and significantly
outperforms the isotropic transmission. By contrast, when $K$ is
large, the isotropic transmission performs close to the optimal transmit
covariance, while the joint beamforming design is observed to lead
to nearly zero data rates. This is consistent with the observations
in the multicast channel for communication only \citep{sidiropoulos2006transmit}.

\begin{figure}[tbh]
\includegraphics[scale=0.5]{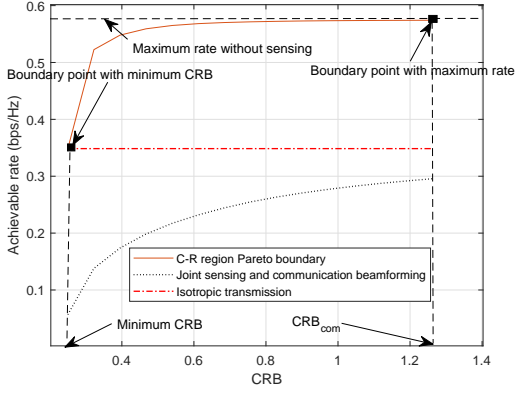}\centering\caption{\label{fig:2}The C-R region with $K=35$.}
\end{figure}

\begin{figure}[tbh]
\includegraphics[scale=0.5]{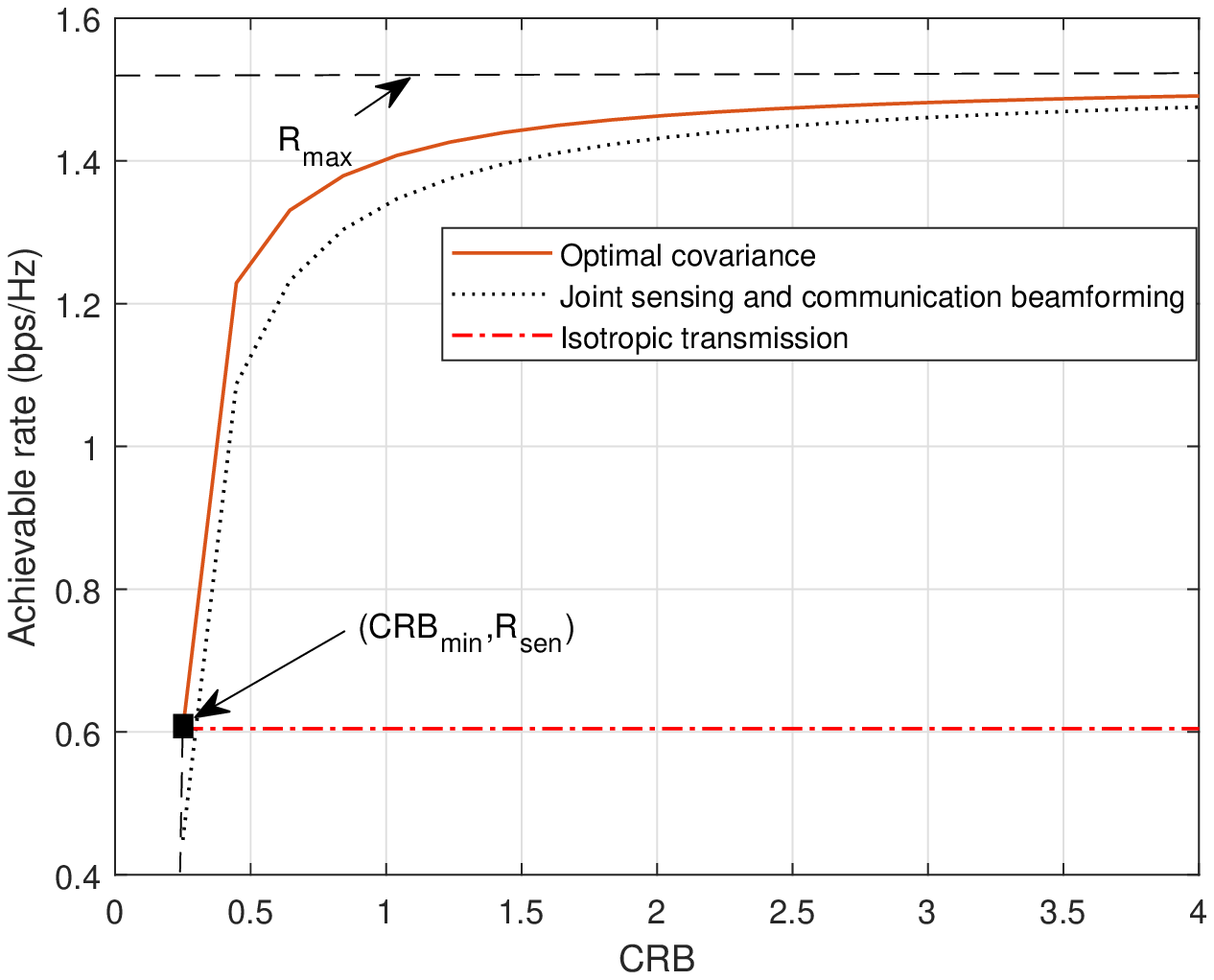}\centering\caption{\label{fig:3}The C-R region with $K=3$.}
\end{figure}

\begin{figure}[tbh]
\includegraphics[scale=0.5]{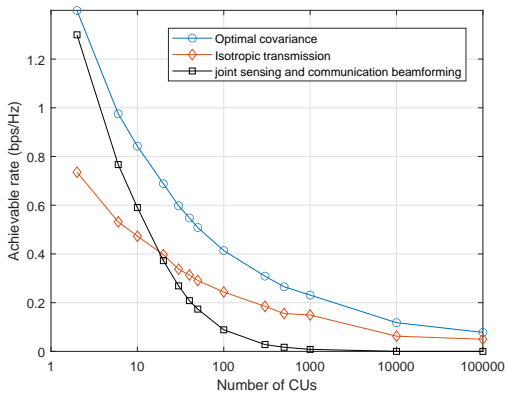}\centering\caption{\label{fig:4}The achievable rate versus the number $K$ of CUs.}
\end{figure}

\section{Conclusion}

This paper studied the fundamental CRB-rate tradeoff in a multi-antenna
multicast channel with ISAC. We characterized the Pareto boundary
of the C-R region via maximizing the communication rate subject to
a new maximum CRB constraint, for which the optimal transmit covariance
solution was presented in semi-closed form. It was shown that optimal
full-rank transmit covariance can be decomposed into two parts over
communication and sensing subchannels, respectively, with different
power allocation strategies. We also presented a practical joint communication
and sensing beamforming. Numerical results were presented to show
the C-R region under the optimal covariance and joint beamforming.
How to extend the design to the cases with multi-antenna CUs and/or
multi-group multicast channels with ISAC is interesting research directions
that are worth pursuing in future.

\appendices{}

\section{Proof of Lemma 1}

First, consider the case with $\lambda_{1}=0$, in which problem (\ref{eq:dual function-1})
is equivalent to minimizing $\mathrm{tr}(\boldsymbol{A}(\lambda_{2},\{\mu_{k}\})\boldsymbol{S}_{x})$.
Recall that $\boldsymbol{A}(\lambda_{2},\{\mu_{k}\})\succeq\boldsymbol{0}$
and $\boldsymbol{S}_{x}\succeq\boldsymbol{0}$. It is thus evident
that any $\boldsymbol{S}_{x}^{*}$ satisfying $\boldsymbol{A}(\lambda_{2},\{\mu_{k}\})\boldsymbol{S}_{x}^{*}=\boldsymbol{0}$
is optimal. 

Next, we consider the case with $\lambda_{1}>0$. As the \ac{evd}
of $\boldsymbol{A}(\lambda_{2},\{\mu_{k}\})$ is $\boldsymbol{A}(\lambda_{2},\{\mu_{k}\})=\boldsymbol{U}\boldsymbol{\Lambda}\boldsymbol{U}^{H}$
with $\boldsymbol{\Lambda}=\mathrm{diag}(\alpha_{1},\dots,\alpha_{N_{t}})$,
we have $\mathrm{tr}(\boldsymbol{A}(\lambda_{2},\{\mu_{k}\})\boldsymbol{S}_{x})=\mathrm{tr}(\boldsymbol{\Lambda}\boldsymbol{U}^{H}\boldsymbol{S}_{x}\boldsymbol{U})$.
We define 
\begin{align}
\tilde{\boldsymbol{S}_{x}}=\boldsymbol{U}^{H}\boldsymbol{S}_{x}\boldsymbol{U}{\text{~or~}}\boldsymbol{S}_{x}=\boldsymbol{U}\tilde{\boldsymbol{S}_{x}}\boldsymbol{U}^{H}.\label{eqn:S_x}
\end{align}
It is easy to verify that $\mathrm{tr}(\boldsymbol{S}_{x}^{-1})=\mathrm{tr}(\tilde{\boldsymbol{S}_{x}}^{-1})$.
Accordingly, it follows that 
\[
\mathrm{tr}(\boldsymbol{A}(\lambda_{2},\{\mu_{k}\})\boldsymbol{S}_{x})+\lambda_{1}\mathrm{tr}(\boldsymbol{S}_{x}^{-1})=\mathrm{tr}(\boldsymbol{\Lambda}\tilde{\boldsymbol{S}_{x}})+\lambda_{1}\mathrm{tr}(\tilde{\boldsymbol{S}_{x}}^{-1}),
\]
and problem (\ref{eq:dual function-1}) can be equivalently expressed
as 
\begin{align}
\underset{\tilde{\boldsymbol{S}}_{x}\succeq\boldsymbol{0}}{\min}\mathrm{tr}(\boldsymbol{\Lambda}\tilde{\boldsymbol{S}_{x}})+\lambda_{1}\mathrm{tr}(\tilde{\boldsymbol{S}_{x}}^{-1}).\label{eqn:proble:Sx:tilde}
\end{align}
Then we have the following lemma. 
\begin{lem}
\label{lemma:appendix} \textup{For a positive semidefinite matrix
$\boldsymbol{B}\in\mathbb{C}^{M\times M}$ with the $(m,n)$-th entry
being $b(m,n),$ it holds that 
\[
\mathrm{tr}(\boldsymbol{B}^{-1})\geq\sum_{i=1}^{M}\frac{1}{b(i,i)},
\]
where the equality is attained if and only if $\boldsymbol{B}$ is
diagonal.} 
\end{lem}
\begin{proof} See \citep[Appendix I]{ohno2004capacity}. \end{proof}
Based on Lemma \ref{lemma:appendix}, it follows that the optimal
solution of $\tilde{\boldsymbol{S}_{x}}$ to problem \eqref{eqn:proble:Sx:tilde}
should be diagonal. As a result, we express $\tilde{\boldsymbol{S}_{x}}=\mathrm{diag}(\tilde{\beta}_{1},\dots,\tilde{\beta}_{N_{t}})$
with $\tilde{\beta}_{i}\ge0,\forall i\in\{1,\ldots,N_{t}\}$. Accordingly,
problem \eqref{eqn:proble:Sx:tilde} is reduced as 
\begin{align}
\underset{{\{\tilde{\beta}_{i}\ge0\}}}{\min}\sum_{i=1}^{N_{t}}\alpha_{i}\tilde{\beta}_{i}+\lambda_{1}\sum_{i=1}^{N_{t}}\frac{1}{\tilde{\beta}_{i}}.\label{eqn:proble:beta}
\end{align}
It is clear that the optimal solution to problem \ref{eqn:proble:beta}
is given by 
\begin{align}
\tilde{\beta}_{i}^{*}\left\{ \begin{array}{ll}
=\lambda_{1}^{1/2}\alpha_{i}^{-1/2} & {\text{if~}}\alpha_{i}>0\\
\rightarrow\infty & {\text{if~}}\alpha_{i}=0
\end{array}\right.,\forall i\in\{1,\ldots,N_{t}\}.\label{eqn:beta:opt}
\end{align}

By combining \eqref{eqn:S_x}, Lemma \ref{lemma:appendix}, and \eqref{eqn:beta:opt},
we have the optimal solution to problem (\ref{eq:dual function-1})
as $\boldsymbol{S}_{x}^{*}=\lambda_{1}^{1/2}\boldsymbol{U}^{H}{\boldsymbol{\Sigma}}\boldsymbol{U},$
where ${\boldsymbol{\Sigma}}=\mathrm{diag}(\beta_{1}^{*},\dots,\beta_{N_{t}}^{*})$
with $\beta_{i}=\tilde{\beta}_{i}^{*}/\lambda_{1}^{1/2},$ $\forall i\in\{1,\ldots,N_{t}\}$.
This thus completes the proof.

{\footnotesize{}{} \bibliographystyle{IEEEtran}
\bibliography{IEEEabrv,IEEEexample,my_ref}

\begin{thebibliography}{10}
\providecommand{\url}[1]{#1}
\csname url@samestyle\endcsname
\providecommand{\newblock}{\relax}
\providecommand{\bibinfo}[2]{#2}
\providecommand{\BIBentrySTDinterwordspacing}{\spaceskip=0pt\relax}
\providecommand{\BIBentryALTinterwordstretchfactor}{4}
\providecommand{\BIBentryALTinterwordspacing}{\spaceskip=\fontdimen2\font plus
\BIBentryALTinterwordstretchfactor\fontdimen3\font minus
  \fontdimen4\font\relax}
\providecommand{\BIBforeignlanguage}[2]{{%
\expandafter\ifx\csname l@#1\endcsname\relax
\typeout{** WARNING: IEEEtran.bst: No hyphenation pattern has been}%
\typeout{** loaded for the language `#1'. Using the pattern for}%
\typeout{** the default language instead.}%
\else
\language=\csname l@#1\endcsname
\fi
#2}}
\providecommand{\BIBdecl}{\relax}
\BIBdecl

\bibitem{liu2021integrated}
F.~Liu, Y.~Cui, C.~Masouros, J.~Xu, T.~X. Han, Y.~C. Eldar, and S.~Buzzi,
  ``Integrated sensing and communications: {Towards} dual-functional wireless
  networks for {6G} and beyond,'' \emph{IEEE J. Sel. Areas Commun.}, vol.~40,
  no.~6, pp. 1728--1767, Jun. 2022.

\bibitem{heath2018foundations}
R.~W. Heath and A.~Lozano, \emph{Foundations of {MIMO Communication}}.\hskip
  1em plus 0.5em minus 0.4em\relax Cambridge University Press, 2018.

\bibitem{haimovich2007mimo}
A.~M. Haimovich, R.~S. Blum, and L.~J. Cimini, ``{MIMO} radar with widely
  separated antennas,'' \emph{IEEE Signal Process. Mag.}, vol.~25, no.~1, pp.
  116--129, Dec. 2007.

\bibitem{LiStoJ07}
J.~Li and P.~Stoica, ``{MIMO} radar with colocated antennas,'' \emph{IEEE
  Signal Process. Mag.}, vol.~24, no.~5, pp. 106--114, Sep. 2007.

\bibitem{LiuZhouJ18}
F.~Liu, L.~Zhou, C.~Masouros, A.~Li, W.~Luo, and A.~Petropulu, ``Toward
  dual-functional radar-communication systems: Optimal waveform design,''
  \emph{{IEEE} Trans. Signal Process.}, vol.~66, no.~16, pp. 4264--4279, Aug.
  2018.

\bibitem{LiuHuangNirJ20}
X.~Liu, T.~Huang, N.~Shlezinger, Y.~Liu, J.~Zhou, and Y.~C. Eldar, ``Joint
  transmit beamforming for multiuser {MIMO} communications and {MIMO} radar,''
  \emph{{IEEE} Trans. Signal Process.}, vol.~68, pp. 3929--3944, Jun. 2020.

\bibitem{hua2021optimal}
H.~Hua, J.~Xu, and T.~X. Han, ``Transmit beamforming optimization for
  integrated sensing and communication,'' in \emph{Proc. IEEE Global Commun.
  Conf. (GLOBECOM)}, 2021, pp. 01--06.

\bibitem{NOMAISAC}
Z.~Wang, Y.~Liu, X.~Mu, Z.~Ding, and O.~A. Dobre, ``{NOMA} empowered integrated
  sensing and communication,'' \emph{IEEE Commun. lett.}, vol.~26, no.~3, pp.
  677--681, Mar. 2022.

\bibitem{yin2022rate}
L.~Yin, Y.~Mao, O.~Dizdar, and B.~Clerckx, ``Rate-splitting multiple access for
  {6G}--{Part II}: Interplay with integrated sensing and communications,''
  \emph{arXiv preprint arXiv:2205.02462}, 2022.

\bibitem{liu2022survey}
A.~Liu, Z.~Huang, M.~Li, Y.~Wan, W.~Li, T.~X. Han, C.~Liu, R.~Du, D.~K.~P. Tan,
  J.~Lu \emph{et~al.}, ``A survey on fundamental limits of integrated sensing
  and communication,'' \emph{IEEE Commun. Surveys Tuts.}, vol.~24, no.~2, pp.
  994--1034, 2nd Quart. 2022.

\bibitem{liu2021cramer}
F.~Liu, Y.-F. Liu, A.~Li, C.~Masouros, and Y.~C. Eldar, ``{Cram{\'e}r-Rao}
  bound optimization for joint radar-communication beamforming,'' \emph{{IEEE}
  Trans. Signal Process.}, vol.~70, pp. 240--253, Dec. 2021.

\bibitem{xiong2022flowing}
Y.~Xiong, F.~Liu, Y.~Cui, W.~Yuan, and T.~X. Han, ``Flowing the information
  from {Shannon to Fisher}: {Towards} the fundamental tradeoff in {ISAC},''
  \emph{arXiv preprint arXiv:2204.06938}, 2022.

\bibitem{Haocheng2022}
H.~Hua, X.~Song, Y.~Fang, T.~X. Han, and J.~Xu, ``{MIMO} integrated sensing and
  communication with extended target: {CRB-rate} tradeoff,'' \emph{arXiv
  preprint arXiv:2205.14050}, 2022.

\bibitem{jindal2006capacity}
N.~Jindal and Z.-Q. Luo, ``Capacity limits of multiple antenna multicast,'' in
  \emph{Proc. IEEE ISIT.}, 2006, pp. 1841--1845.

\bibitem{StoPETLiJ07}
P.~Stoica, J.~Li, and Y.~Xie, ``On probing signal design for {MIMO} radar,''
  \emph{{IEEE} Trans. Signal Process.}, vol.~55, no.~8, pp. 4151--4161, Aug.
  2007.

\bibitem{kay1993fundamentals}
S.~M. Kay, \emph{Fundamentals of {Statistical Signal Processing: Estimation
  Theory}}.\hskip 1em plus 0.5em minus 0.4em\relax Prentice-Hall, Inc., 1993.

\bibitem{boyd2004convex}
S.~Boyd and L.~Vandenberghe, \emph{Convex {O}ptimization}.\hskip 1em plus 0.5em
  minus 0.4em\relax Cambridge {U}niversity {P}ress, 2004.

\bibitem{boyd2014ellipsoid}
\BIBentryALTinterwordspacing
S.~Boyd, ``Ellipsoid method,'' May. 2014. [Online]. Available:
  \url{https://web.stanford.edu/class/ee364b/lectures/ellipsoid_method_notes.pdf}
\BIBentrySTDinterwordspacing

\bibitem{mohseni2006optimized}
M.~Mohseni, R.~Zhang, and J.~M. Cioffi, ``Optimized transmission for fading
  multiple-access and broadcast channels with multiple antennas,'' \emph{IEEE
  J. Sel. Areas Commun.}, vol.~24, no.~8, pp. 1627--1639, Aug. 2006.

\bibitem{cvx}
\BIBentryALTinterwordspacing
M.~Grant and S.~Boyd, ``{CVX}: Matlab software for disciplined convex
  programming, version 2.1,'' Mar. 2014. [Online]. Available:
  \url{http://cvxr.com/cvx}
\BIBentrySTDinterwordspacing

\bibitem{sidiropoulos2006transmit}
N.~D. Sidiropoulos, T.~N. Davidson, and Z.-Q. Luo, ``Transmit beamforming for
  physical-layer multicasting,'' \emph{{IEEE} Trans. Signal Process.}, vol.~54,
  no.~6, pp. 2239--2251, Jun. 2006.

\bibitem{ohno2004capacity}
S.~Ohno and G.~B. Giannakis, ``Capacity maximizing {MMSE-optimal} pilots for
  wireless {OFDM} over frequency-selective block {Rayleigh-fading} channels,''
  \emph{{IEEE} Trans. Inf. Theory}, vol.~50, no.~9, pp. 2138--2145, Sep. 2004.

\end{thebibliography}
}{\footnotesize\par}
\end{document}